\documentclass[submission,copyright,creativecommons]{eptcs}
\providecommand{\event}{NCMA 2024} 

\usepackage{iftex}
\usepackage{enumerate}
\usepackage{amssymb}
\usepackage{amsmath}
\usepackage{amsthm}
\usepackage{mathrsfs}
\usepackage{mathtools}
\usepackage{leftidx}
\usepackage{paralist}
\usepackage{float}
\usepackage[sort&compress,numbers,sectionbib]{natbib}
\usepackage{pgf}
\usepackage{tikz}
\usetikzlibrary{arrows}
\newcommand{\where}{\mid} 
\newcommand{\eps}{\varepsilon} 
\renewcommand{\phi}{\varphi} 
\renewcommand{\rho}{\varrho} 
\newcommand{\ra}{\rightarrow}      
\newcommand{\Ra}{\Rightarrow} 

\DeclareMathOperator{\aalph}{alph}
\DeclareMathOperator{\card}{card}

\DeclareMathOperator{\inD}{in-d}
\DeclareMathOperator{\outD}{out-d}
\DeclareMathOperator{\frontier}{frontier}
\newcommand{\lfam}[1]{\mathbf{#1}}
\newcommand{\CF}{\lfam{CF}}   
\newcommand{\CS}{\lfam{CS}}   
\newcommand{\RE}{\lfam{RE}}   

\newcommand{\itparagraph}[1]{\par\medskip\par\noindent\emph{#1. }\ignorespaces}
\newcommand{\ibasis}{\itparagraph{Basis}}
\newcommand{\ihypothesis}{\itparagraph{Induction Hypothesis}}
\newcommand{\istep}{\itparagraph{Induction Step}}

\newtheorem{nclaim}{Claim}{\it}{\rm}
\newtheorem{ndefinition}{Definition}{\bf}{\rm}
\newtheorem{example}{Example}
\newtheorem{theorem}{Theorem}
\newtheorem{lemma}{Lemma}

\ifpdf
  \usepackage{underscore}         
  \usepackage[T1]{fontenc}        
\else
  \usepackage{breakurl}           
\fi

\title{How to Demonstrate Metalinearness and Regularity\\ by Tree-Restricted General Grammars}
\author{Martin Havel
\institute{Brno University of Technology,\\ Faculty of Information Technology,\\ Brno, Czech republic}
\email{ihavelm@fit.vut.cz}
\and
Zbyněk Křivka
\institute{Brno University of Technology, \\ Faculty of Information Technology,\\Brno, Czech republic}
\email{krivka@fit.vut.cz}
\and
Alexander Meduna
\institute{Brno University of Technology, \\ Faculty of Information Technology,\\Brno, Czech republic}
\email{meduna@fit.vut.cz}
}
\def\titlerunning{How to Demonstrate Metalinearness and Regularity by Tree-Restricted General Grammars}
\def\authorrunning{M. Havel, Z. Křivka \& A. Meduna}
\begin{document}
\maketitle

\begin{abstract}
This paper introduces derivation trees for general grammars. Within these trees, it defines context-dependent pairs of nodes, corresponding to rewriting two neighboring symbols using a non context-free rule. It proves that the language generated by a linear core general grammar with a slow-branching derivation tree is $k$-linear if there is a constant $u$ such that every sentence $w$ in the generated language is the frontier of a derivation tree in which any pair of neighboring paths contains $u$ or fewer context-dependent pairs of nodes. Next, it proves that the language generated by a general grammar with a regular core is regular if there is a constant $u$ such that every sentence $w$ in the generated language is the frontier of a derivation tree in which any pair of neighboring paths contains $u$ or fewer context-dependent pairs
of nodes. The paper explains that this result is a powerful tool for showing that certain languages are $k$-linear or
regular.
\end{abstract}

\section{Introduction}
\label{sec:Introduction}
Formal language theory has always intensively struggled to establish conditions under
which general grammars generate a proper subfamily of the family of recursively enumerable languages
because results like this often significantly simplify proofs that some languages are members of the subfamily.
Continuing with this important investigation trend in formal language theory, the present paper establishes another result of this kind based upon a restriction illustrated in Fig. \ref{graph:introduction} placed upon a graph-based representation of derivations in general grammars. 

Concerning general grammars, which generate a proper subfamily of the family of recursively enumerable languages, some results of this kind have been achieved, too. First of all, \cite{Matt64} states that for a grammar, the set of terminal strings generated by left-to-right derivations is context-free. Second, \cite{Matt67} shows that the set of terminal strings generated by two-way derivations is context-free, which is further studied in \cite{Boo73}. Third, \cite{Boo72} demonstrates that a grammar generates a context-free language if the left-hand side of every rule contains only one nonterminal with terminal strings as the only context. Fourth, also \cite{Boo72} shows that if every rule of a general grammar has as its left context a string of terminal symbols at least as long as the right context, then the generated language is context-free. Fifth, \cite{Bak74} demonstrates that a grammar generates a context-free language if the right-hand side of every rule contains a string of terminals longer than any string of terminals between two nonterminals on the left-hand side. For $k$-linear grammar, there is no such study. For regularity, there is the publication \cite{Ehr83}, which shows regularity only in context-free languages.

Finally, Section 2.3.2 in \cite{MEDSOU} demonstrates context-freeness based on the tree restriction with context-dependency. We explain and expand the importance of introduced context-dependency (see Fig. \ref{graph:introduction}) to demonstrate metalinearness and regularity.

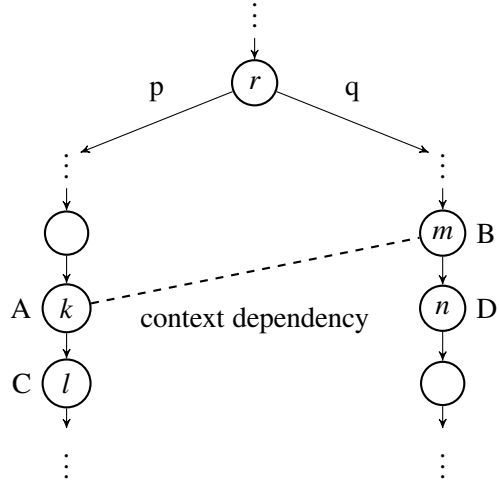
\begin{figure}[ht]
\label{graph:introduction}
\centering
\tikzset{
  treenode/.style = {align=center, inner sep=0pt, text centered, font=\sffamily},
  node/.style = {treenode, circle, text width=1.5em, draw=black, thick}
}
\begin{tikzpicture}[->,>=stealth',level/.style={sibling distance = 10cm/#1, level distance = 1cm}]
\tikzstyle{level 1}=[sibling distance=5cm]
\node {$\vdots$}
    child{ node [node,label={[shift={(-1.3,-0.7)}]p},label={[shift={(1.3,-0.7)}]q}] {$r$}
        child{ node {$\vdots$}
            child{ node [node] {} 
                child{ node (k) [node,label=left:{A}] {$k$}
                    child{ node [node,label=left:{C}] {$l$}
                        child{ node {$\vdots$}}
                    }
                }
            }
        }
        child{ node {$\vdots$}
            child{ node (m) [node,label=right:{B}] {$m$}
                child{ node [node,,label=right:{D}] {$n$}
                    child{ node [node] {}
                        child{ node {$\vdots$}}
                    }
                }
            }
        }
    }
;
\draw [dashed,thick,-] (k) -- (m) node[midway,below=10] {context dependency};
\end{tikzpicture}
\caption{Illustration of context dependency in $t$}
\end{figure}

To give an insight into the new result achieved in the present paper, some terminology is first needed to be sketched. We introduce a linear core general grammar $G$
if any $p \in P$ has one of these forms,
\[
    AB \ra CD\textnormal{, }A \ra BC\textnormal{, }A \ra xEy
\]
where $A,B,C,D$ are nonterminals, $E$ is a nonterminal or the empty string, and $x, y$ are strings of terminals.

We define the notion of a derivation tree $t$ graphically representing a derivation in $G$ by analogy with this notion in terms of a $k$-linear grammar (see Definiton \ref{def:k-lin} or Section 6.2 in \cite{Roz97}).
However, in addition, we introduce context-dependent pairs of nodes in $t$ as follows.
In $t$, two paths are neighboring if no other path occurs between them.
Let $p$ and $q$ be two neighboring paths in $t$. Let $p$ contain a node $k$ with a single child $l$,
where $k$ and $l$ are labeled with $A$ and $C$, respectively, and let $q$ contain a node $m$ with a single child $n$,
where $m$ and $n$ are labeled with $B$ and $D$, respectively.
Let this four-node portion of $t$; consisting of $k$, $l$, $m$, and $n$;
graphically represents an application of $AB \ra CD$.
Then, $k$ and $m$ are a context-dependent pair of nodes (see Fig. \ref{graph:introduction}).

The main theorem provided in this paper represents a powerful tool to demonstrate that if a linear core general grammar $H$ generates each of its sentences by a derivation satisfying prescribed conditions
(specifically, one of these conditions requires that there is a positive integer $u$ and any two nonterminal neighboring paths contain no more than $u$ pairs of context-dependent nodes) then the language generated by $H$ is $k$-linear. Similarly, the following theorems provide a tool to demonstrate membership in the regular language family.

\section{Preliminaries}
\label{sec:Preliminaries}
We assume that the reader is familiar with graph theory, including labeled ordered trees and their terminology
(see~\cite{AU, CorLeiRiv2002, MAH}) as well as formal language theory (see~\cite{FlaC,FL,Roz97}).

A directed graph $G$ is a pair $G = (V,E)$, where $V$ is a finite set of nodes
and $E \subseteq V^2$ is a finite set of edges. For a node $v \in V$
a number of edges of the form $(x,v) \in E$ and a number of edges of the form $(v,y) \in E$, for $x,y \in V$,
is called an in-degree of $v$ and an out-degree of $v$, respectively, and denoted by $\inD(v)$, $\outD(v)$.
Let $(v_1,v_2,\dots,v_n)$ be an $n$-tuple of nodes, for some $n \geq 1$, where $v_i \in V$, for $1 \leq i \leq n$,
and there exists an edge $(v_k,v_{k+1}) \in E$, for every pair of nodes $v_k$,$v_{k+1}$, where $1 \leq k \leq n-1$,
then, we call it a sequence of length $n$.
Let $(v_1,v_2,\dots,v_n)$ be a sequence of the length $n$, for some $n \geq 1$,
where $v_i \neq v_j$, for $1 \leq i \leq n$, $1 \leq j \leq n$, then, we call the sequence a path.
Let $(v_1,v_2,\dots,v_n)$ be a path in $G$, for some $n \geq 1$, and $v_1 = v_n$, then we call it a cycle.
A graph $G$ is acyclic iff it contains no cycle.

For a set~$W$, $\card(W)$ denotes its cardinality.
Let $V$ be an alphabet (finite non-empty set). $V^*$ is the set of all strings over V.
Algebraically, $V^*$ represents the free monoid generated by~$V$ under the operation
of concatenation. The unit of~$V^*$ is denoted by~$\eps$.
Set $V^+ = V^* - \{\eps\}$. Algebraically, $V^+$ is thus the free semigroup
generated by~$V$ under the operation of concatenation.
For~$w \in V^*$, $|w|$ denotes the length of~$w$.
The alphabet of $w$, denoted by $\aalph(w)$, is the set of symbols appearing in $w$.
Let $\mathcal{I}$ denote the set of all positive integers.

Let $\Ra$ be a relation over $V^*$. The transitive and transitive and reflexive
closure of $\Ra$ are denoted $\Ra^+$ and $\Ra^*$, respectively.
Unless explicitly stated otherwise, we write $x\ \Ra\ y$ instead of $(x,y) \in\ \Ra$.

The families of context-free, context-sensitive and recursively enumerable languages
are denoted by~$\CF$, $\CS$ and $\RE$, respectively. 

\section{Definitions and Examples}
\label{sec:definitions}

\begin{ndefinition}
\label{Tree}
An \emph{(oriented) tree} is a directed acyclic graph $G = (V,E)$,
with a specified node $r \in V$ called the \emph{root} such that $\inD(r) = 0$, $\inD(x) = 1$,
and there exists a path $(v_1,v_2,\dots,v_n)$, where $v_1 = r$, $v_n = x$, for some $n \geq 1$,
for all $x \in V - \{r\}$. For $v,u \in V$, where $(v,u) \in E$, $v$ is called a \emph{parent} of $u$,
$u$ is called a child of $v$, respectively. For $v,u,z \in V$, where $(v,u),(v,z) \in E$,
$u$ is called a \emph{sibling} of $z$ and vice versa.

A tree is called \emph{labeled}, if there exist a set of labels $\mathcal{L}$ and a total mapping $l: V \ra \mathcal{L}$.

An \emph{ordered tree} $t$ is a tree, where for every set of siblings there exists a linear ordering.
Let $o$ has the children $n_1$,$n_2$,\dots,$n_r$ ordered in this way, where $r \geq 1$.
Then $n_1$ is the \emph{leftmost child} of $o$, $n_r$ is the \emph{rightmost child} of $o$
and $n_i$ is the \emph{direct left sibling} of $n_{i+1}$, $n_{i+1}$ is the \emph{direct right sibling} of $n_i$,
$1 \leq i \leq r - 1$, and for $j < k$,
$n_j$ is \emph{left sibling} of $n_k$ and $n_k$ is \emph{right sibling} of $n_j$,
$1 \leq j \leq r$, $1 \leq k \leq r$.

Let $t$ be a labeled ordered tree, and let $t$ contain node $o$.
Let $\alpha = (o,m_1,m_2,\dots,m_r)$, and $\beta = (o,n_1,n_2,\dots,n_s)$
be two paths in $t$, for some $r,s \geq 1$, such that $o$ is the parent of $m_1$ and $n_1$, while
\begin{enumerate}
    \item $m_1$ is the direct left sibling of $n_1$;
    \item $m_i$ is a nonterminal child of $m_{i-1}$, while all its right siblings are terminal siblings,
          $2 \leq i \leq r - 1$,
          $n_j$ is a nonterminal child of $n_{j-1}$, while all its left siblings are terminal siblings,
          $2 \leq j \leq s - 1$;
    \item if $m_r$ is a terminal node, then all its siblings are terminal nodes;
          otherwise, all its right siblings are terminal siblings;
    \item if $n_s$ is a terminal node, then all its siblings are terminal nodes;
          otherwise, all its left siblings are terminal siblings;
\end{enumerate}
Then, $\alpha$ and $\beta$ are two \emph{nonterminal neighboring paths} in $t$,
$\alpha$ is a \emph{left nonterminal neighboring path} to $\beta$,
and $\beta$ is a \emph{right nonterminal neighboring path} to $\alpha$.

\end{ndefinition}

Next, we define the notion of a general grammar, also known as that of a type-0 grammar or that of a phrase-structure grammar in the literature.

\begin{ndefinition}
\label{def:lc}
A \emph{general grammar} (GG) $G$ is a quadruple $G = (V$, $T$, $P$, $S)$,
where $V$ is a \emph{total alphabet}, $T \subset V$ is a \emph{terminal alphabet},
$P$ is a finite \emph{set of rules} of the form $x \ra y$,
where $x,y \in V^*$, $\aalph(x) \cap (V - T) \neq \emptyset$,
$S \in V - T$ is a \emph{start symbol}.
For every $u,v \in V^*$ and $x \ra y \in P$, $uxv \Ra uyv[p]$ or simply $uxv \Ra uyv$
is a \emph{derivation step} of $G$ from $uxv$ to $uyv$ by the rule $x \ra y$,
$\Ra$ is the \emph{direct derivation} relation.
Let $w_0,w_1,\dots,w_n \in V^*$, for some $n \geq 1$, such that $w_0 \Ra w_1\,[p1] \Ra \dots \Ra w_n\,[p_n]$,
where $p_i \in P$, for all $i = 1,\dots,n$, then, $w_0 \Ra^n w_n$;
based on $\Ra^n$, we define $\Ra^+$ and $\Ra^*$.

A \emph{language} of $G$ is $L(G) = \{w \in T^* \where S \Ra^* w\}$.
$G$ is \emph{propagating} if $A \ra x \in P$ implies $x \neq \eps$.
$G$ is \emph{context-free} if $A \ra x \in P$ implies $A \in V - T$.
$G$ is \emph{linear core} GG if any $p \in P$ has one of these forms:
\[
    AB \ra CD\textnormal{, }A \ra BC\textnormal{, }A \ra xEy
\]
where $A,B,C,D \in V - T$, $E \in (V - T) \cup \{\eps\}$, $x,y \in T^*$.
In what follows, unless explicitly stated otherwise, we assume that every GG is a linear core GG.\\
Similary, $G$ is \emph{left linear core} GG if any $p \in P$ has one of these forms (see \cite{Kur}):
\[
    AB \ra CD\textnormal{, }A \ra BC\textnormal{, }A \ra xE
\]
where $A,B,C,D \in V - T$, $E \in (V - T) \cup \{\eps\}$, $x \in T^*$.\\
 $G$ is GG in the \emph{Kuroda normal form} (KNF) \cite{Med20} if every rule is one of these forms:
\[
    AB  \ra CD\textnormal{, }A \ra BC\textnormal{, }A \ra B\textnormal{, }A \ra a\textnormal{, }A \ra \varepsilon
\]
where $A,B,C,D \in V - T$, $a \in T$.
\end{ndefinition}
As obvious, all rules of the form of $A \ra B$ can be always removed from $G$ without disturbing $L(G)$. Next, we show that the proposed linear core grammars have the same generative power as GGs.
\begin{lemma}
\label{th:GG:RE}
A language $L$ is recursively enumerable iff $L = L(G)$, where $G$ is a linear core general grammar.
\end{lemma}
\begin{proof}
Every language $L$ generated by a linear core GG $G$ is recursively enumerable, because every general linear core grammar can be trivially converted to KNF. In other direction, every KNF $G$ is a linear core GG by Definition \ref{def:lc}.

\end{proof}

\begin{lemma}
\label{th:pGG:CS}
A language $L$ is context-sensitive iff $L = L(G)$,
where $G$ is a propagating linear core general grammar.
\end{lemma}

\begin{proof}
Every language $L$ generated by propagating linear core GG $G$ is context sensitive, because each rule in $P$, where $P \in G$, is a form of $x \rightarrow y$ and $|x| \leq |y|$.

\end{proof}
\begin{ndefinition}
\label{def:k-lin}
A \emph{linear grammar} $G$ is a GG $G = (V$, $T$, $P$, $S)$,
where $P$ contains rules of the form:
\[
    A \ra x
\]
where $A \in (V-T)$, $x \in T^*((V-T) \cup \{\varepsilon\})T^*$.
A language is \emph{linear} (1-linear) if it can be generated by a linear grammar.
The concept of a linear grammar can be generalized:
A \emph{$k$-linear grammar} $G$ is a GG $G = (V$, $T$, $P$, $S)$,
where $P$ is a finite set of rules of  the form:
\[
A \ra x\textnormal{, }A \ra xBy\textnormal{, }S \ra W
\]
where $A,B \in (V-T)$, $x,y \in T^*$, $W \in (V-(T \cup \{S\}))^k$.
A language is said to be \emph{$k$-linear} if it can be generated by a $k$-linear grammar.
A language is said to be \emph{metalinear} if it is $k$-linear for some positive integer $k$.

\end{ndefinition}
\begin{ndefinition}
Let $G = (V, T, P, S)$ be a linear core GG without rules of the form $AB \ra CD$.
Let $w \in T^*$ be a string derived from $G$.
A derivation tree for $w$ is a labeled tree $\tau$ such that:
\begin{enumerate}
    \item The root of $\tau$ is labeled with $S$.
    \item Each leaf of $\tau$ is labeled with a symbol from $T$.
    \item Each internal node of $\tau$ is labeled with a symbol from $V$.
    \item If an internal node $v$ is labeled with $A \in V$ and has children labeled $B_{1}$, $B_{2}$, then there exists a rule $A \ra B_{1}B_{2}$ in $P$ and, analogically, for the rest of the rules of a linear core GG without rules of the form $AB \ra CD$.
    \item The yield of $\tau$ (that is, the concatenation of the labels on its leaves), denoted by $\frontier(\tau)$, is $w$.
\end{enumerate}

\end{ndefinition}

\begin{example}
\label{ex:labeled:ordered:tree}
The following graph (Fig. \ref{graph:KNFtree}) represents a labeled ordered tree $t$ for a GG in KNF.
Since any two distinct nodes have different labels, we refer to their
labels below.
The root node $\hat{r}$ is $a$. It has no parent and two children $b$ and $c$.
Then $b$ is a sibling of $c$ and $c$ is a sibling of $b$.
The leftmost child of $b$ is $d$, while the rightmost is $e$.
The node $d$ is a left sibling of $e$.
The node $d$ is the parent of $h$, but $h$ has no child, so it is a leaf node.
$horsm = \frontier(t)$. Consider the node $e$.
The nodes $a$ and $b$ are predecessors of $e$,
while $i$, $j$, $o$, $p$, and $r$ are descendants of $e$.
The nodes $c$ or $d$ are not in predecessor relation with $e$,
as they are neither predecessors of $e$, nor descendants of $e$.
The sequence of nodes $bejpr$ is a path in $t$.
The path $cglqs$ is neighboring to $bejpr$; unlike $acglqs$, $eio$ or $bdh$.

\begin{figure}[H]

\centering
\tikzset{
  treenode/.style = {align=center, inner sep=0pt, text centered, font=\sffamily},
  node/.style = {treenode, circle, text width=1.5em, draw=black, thick}
}
\begin{tikzpicture}[->,>=stealth',level/.style={sibling distance = 4cm/#1, level distance = 0.8cm}]
\tikzstyle{level 1}=[sibling distance=5cm]
\node [node] {$a$}
    child{ node [node] {$b$} 
        child{ node [node] {$d$} 
            child{ node [node] {$h$}}
        }
        child{ node [node] {$e$}
            child{ node [node] {$i$}
                child{ node [node] {$o$}}
            }
            child{ node [node] {$j$}
                child{ node [node] {$p$}
                    child{ node [node] {$r$}}
                }
            }
        }
    }
    child{ node [node] {$c$}
        child{ node [node] {$g$}
            child{ node [node] {$l$}
                child{ node [node] {$q$}
                    child{ node [node] {$s$}}
                }
            }
            child{ node [node] {$m$}}
        }
    }
;
\end{tikzpicture}
\caption{Labeled ordered tree $t$}
\label{graph:KNFtree}
\end{figure}
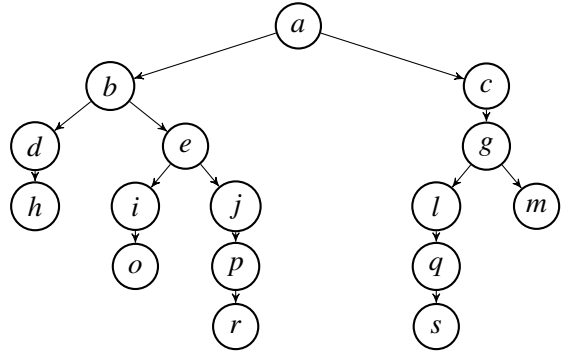
\end{example}

\begin{ndefinition}
Let $G = (V,T,P,S)$ be a linear core GG.
\begin{enumerate}
    \item For $p\colon A \ra x \in P$, $A\langle x\rangle$ is the \emph{rule tree} that represents $p$.
    \item\label{DT:2} The \emph{derivation trees} representing the derivations in $G$ are defined recursively as follows:
    \begin{enumerate}
        \item One-node tree with a node labeled $X$ is the derivation tree corresponding to $X \Ra^0 X$ in $G$,
              where $X \in V$. If $X = \eps$, we refer to the node  
              labeled $X$ as \emph{$\eps$-node} (\emph{$\eps$-leaf}); otherwise, we call it
              \emph{non-$\eps$-node} (\emph{non-$\eps$-leaf}).
        \item Let $d$ be the derivation tree with $\frontier(d) = uAv$ representing $X \Ra^* uAv\,[\rho]$
              and let $p\colon A \ra x \in P$. The derivation tree that represents
              \[
                X \Ra^* uAv\,[\rho] \Ra uxv\,[p]
              \]
              is obtained by replacing the $i$th non-$\eps$-leaf in $d$ labeled $A$, with the rule tree corresponding to $p$, $A\langle x\rangle$, where $i = |uA|$.
        \item\label{DT:2:3}
              Let $d$ be the derivation tree with $\frontier(d) = uABv$ representing $X \Ra^* uABv\,[\rho]$
              and let $p\colon AB \ra CD \in P$. The derivation tree that represents
              \[
                X \Ra^* uABv\,[\rho] \Ra uCDv\,[p]
              \]
              is obtained by replacing the $i$th and $(i+1)$th non-$\eps$-leaf in $d$ labeled $A$ and $B$
              with $A\langle C\rangle$ and $B\langle D\rangle$, respectively, where $i = |uA|$.
    \end{enumerate}
    \item A \emph{derivation tree} in $G$ is any tree $t$
          for which there is a derivation represented by $t$ (see item \ref{DT:2} in this definition).
\end{enumerate}

Note that the figure to illustrate the definition is postponed to Example \ref{ex:GG}. Moreover, after replacement in \ref{DT:2:3},
the nodes $A$ and $B$ are the parents of the new leaves $C$ and $D$, respectively,
and we say that $A$ and $B$ are \emph{context-dependent}, alternatively speaking, we say
that there is a context dependency between $A$ and $B$.
In a derivation tree, two nodes are \emph{context-independent} if they are not context-dependent.

Then, for any $p\colon A \ra x \in P$, $\prescript{}{G}\triangle(p)$ denotes the rule tree corresponding to $p$.
For any $A \Ra^* x\,[\rho]$ in $G$, where $A \in V-T$, $x \in V^*$, and $\rho \in P^*$,
$\prescript{}{G}\triangle(A \Ra^* x\,[\rho])$ denotes one of the derivation trees corresponding to $A \Ra^* x\,[\rho]$.
Just like we often write $A \Ra^* x$ instead of $A \Ra^* x\,[\rho]$, we sometimes simplify
$\prescript{}{G}\triangle(A \Ra^* x\,[\rho])$ to $\prescript{}{G}\triangle(A \Ra^* x)$ in what follows
if there is no danger of confusion.
Let $\prescript{}{G}\blacktriangle$ denote the set of all derivation trees in $G$.
Finally, by $\prescript{}{G}\triangle_x \in \prescript{}{G}\blacktriangle$,
we mean a derivation tree whose frontier is $x$, where $x \in L(G)$.

If a node is labeled with a terminal, it is called a \emph{terminal node}.
If a node is labeled with a nonterminal, it is called a \emph{nonterminal node}.
Analogously, we define the notions of a \emph{terminal child}, \emph{nonterminal child},
\emph{terminal sibling}, \emph{nonterminal sibling}. 
If a node is labeled with a nonterminal and has two \emph{nonterminal node} children, it is called a \emph{branching nonterminal node}.
Let $\alpha = (o,m_1,m_2,\dots,m_r)$ and $\beta = (o,n_1,n_2,\dots,n_s)$ be two neighboring paths,
where $r,s \geq 0$, $\alpha$ is the left neighboring path to $\beta$, and $m_r$ and $n_s$ are terminal nodes.
Then, there is a $t$-tuple $\gamma = (g_1,g_2,\dots,g_t)$ of nodes from $\alpha$
and $t$-tuple $\delta = (h_1,h_2,\dots,h_t)$ of nodes from $\beta$,
where $g_p < g_q$, for $1 \leq p < q \leq t$, $t < \min(r,s)$,
and $g_i$ and $h_i$ are context-dependent, for $1 \leq i \leq t$.
Let $\rho = p_1p_2\dots p_t$ be a string of non-context-free rules
corresponding to context dependencies between $\gamma$ and $\delta$.
We call $\rho$ the \emph{right context of $\alpha$} and the \emph{left context of $\beta$} or
the \emph{context of $\alpha$ and $\beta$}.
Consider a node $m_i \in \alpha$, where $1 \leq i \leq r$,
and two $(t-k+1)$-tuples of nodes $\sigma = (g_{k}, g_{k+1},\dots, g_t)$ and $\phi = (h_{k}, h_{k+1},\dots, h_t)$,
where $k$ is a minimal integer such that $m_i < g_{k}$.
Then, a string of non-context-free rules $\tau = p_kp_{k+1}\dots p_t$
corresponding to context dependencies between $\sigma$ and $\phi$
is called the \emph{right descendant context of $m_i$}, for some $1 \leq k \leq t$.
Analogously, we define the notion of the \emph{left descendant context} of a node $n_j$ in $\beta$,
for some $1 \leq j \leq s$.
\end{ndefinition}

\begin{ndefinition}
A labeled ordered tree $t$ is \emph{slow-branching} if any of its pairs of nonterminal neighboring paths
contains no more than two nonterminal nodes having two nonterminal children
and there is no reachable terminal node from nodes of the path between the root and any branching nonterminal node. 
A slow-branching labeled ordered tree is of \emph{degree $k$}
if it contains $k$ branching nonterminal nodes, $k \geq 1$.
\end{ndefinition}

\begin{example}
\label{ex:GG}
Let $G = (V$, $T$, $P$, $S)$ be a GG,
where $V = N \cup T$ such that $N = \{S$, $X$, $Y$, $Z$, $A_1$, $A_2$, $B$, $C_1$, $C_2$,
                $D_1$, $D_2$, $E$, $F_1$, $F_2\}$,
$T = \{a$, $b$, $c$, $0$, $1\}$, and $P$ contains the following rules:

\begin{center}
\begin{tabular}[H]{c c}
\parbox{4cm}{
    \begin{enumerate}[(1)]
        \item\label{rule:1} $S \ra A_1X$ 
        \item\label{rule:2} $X \ra A_2Y$
        \item\label{rule:3} $Y \ra BZ$
        \item\label{rule:4} $Z \ra C_1C_2$
        \item\label{rule:5} $A_1 \ra aA_1$
        \item\label{rule:6} $A_2 \ra A_2a$
        \item\label{rule:7} $B \ra bBc$
        \item\label{rule:8} $C_1 \ra aC_1$
        \item\label{rule:9} $C_2 \ra C_2b$
        \item\label{rule:10} $A_1A_2 \ra D_1D_2$
        \item\label{rule:11} $B \ra E$
    \end{enumerate}
}
&
\parbox{4cm}{
    \begin{enumerate}[(1)]
    \setcounter{enumi}{11}
        \item\label{rule:12} $C_1C_2 \ra F_1F_2$
        \item\label{rule:13} $D_1 \ra 0D_1$
        \item\label{rule:14} $D_2 \ra D_21$
        \item\label{rule:15} $E \ra 0E1$
        \item\label{rule:16} $F_1 \ra 0F_1$
        \item\label{rule:17} $F_2 \ra F_21$
        \item\label{rule:18} $D_1 \ra \varepsilon$
        \item\label{rule:19} $D_2 \ra \varepsilon$
        \item\label{rule:20} $E \ra \varepsilon$
        \item\label{rule:21} $F_1 \ra \varepsilon$
        \item\label{rule:22} $F_2 \ra \varepsilon$ 
    \end{enumerate}
}
\end{tabular}
\end{center}

A graph representing $\prescript{}{G}\triangle(S \Ra^* aaa0011a0011b)$ is illustrated in Fig. \ref{graph:2} and illustrate slow-branch\-ing\-ness. The graph is slow-branching since it has exactly $k$ branching nodes. Those are $S,X,Y,Z$. That any of its pairs of nonterminal neighboring paths
contains no more than two nonterminal nodes having two nonterminal children
and there is no reachable terminal node from nodes of the path between the root and any branching nonterminal node.
Observe that terminal nodes, denoted by square, do not influence any condition.

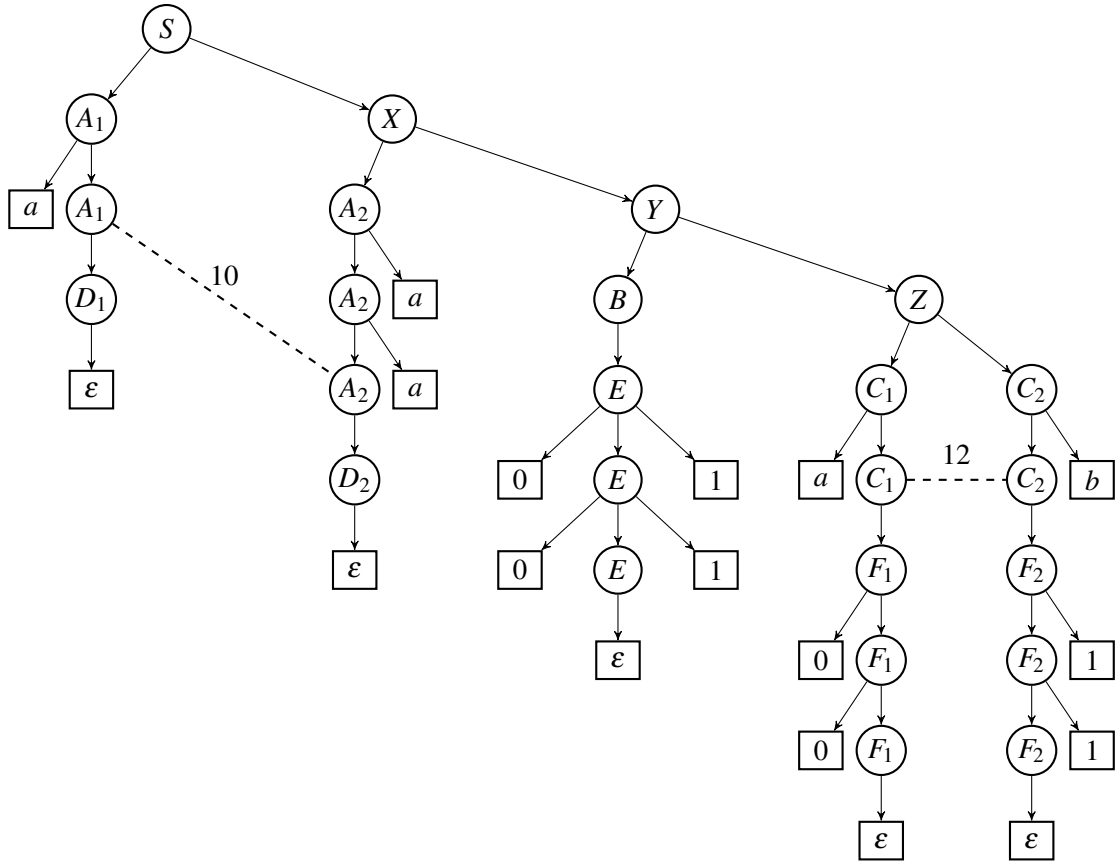
\begin{figure}[ht]
\hspace*{0.5cm}
\tikzset{
  treenode/.style = {align=left, inner sep=0pt, text centered, font=\sffamily},
  node/.style = {treenode, text width=1.5 em, draw=black, thick}
}
\begin{tikzpicture}[->,>=stealth',level/.style={sibling distance = 2.5cm/#1, level distance = 1.2cm}]
\tikzstyle{level 1}=[sibling distance=2cm]
\tikzstyle{level 2}=[sibling distance=1cm]
\tikzstyle{level 3}=[sibling distance=1cm]
\tikzstyle{level 4}=[sibling distance=1cm]
\tikzstyle{level 5}=[sibling distance=1cm]
\tikzstyle{level 6}=[sibling distance=1cm]
\tikzstyle{level 7}=[sibling distance=1cm]
\tikzstyle{level 8}=[sibling distance=1cm]
\tikzstyle{level 9}=[sibling distance=1cm]
\tikzstyle{level 10}=[sibling distance=1cm]

\node [circle, node]  {$S$}
        child{ node   [circle, node] {$A_1$} 
            child[left]{ node [minimum height=1.3em, node] {$a$} }
            child{ node[xshift=-0.5cm] (a1) [circle, node] {$A_1$}     
                child{ node [circle, node] {$D_1$}
                    child{ node [minimum height=1.3em, node] {$\varepsilon$}}
                }
            }
        }
        child{ node[circle, xshift=2cm]  [node] {$X$}
            child{ node  [circle, node] {$A_2$}
            child{ node[circle, xshift=0.5cm] [node] {$A_2$}    
                child{ node[circle, xshift=0.5cm] (a2) [node] {$A_2$}    
                    child{ node [circle, node] {$D_2$} 
                        child{ node [minimum height=1.3em, node] {$\varepsilon$}}
                    }                
                }
                child[right]{ node [minimum height=1.3em, node] {$a$} }
            }
            child[right]{ node [minimum height=1.3em, node] {$a$} }
        }
        child{ node[circle, xshift=3cm] [node] {$Y$} 
        child{ node [circle, node] {$B$}
        child{ node [circle, node] {$E$}         
            child[left]{ node [minimum height=1.3em, node] {$0$} }
            child{ node [circle, node] {$E$}  
                child[left]{ node [minimum height=1.3em, node] {$0$} }
                child{ node [circle, node] {$E$}  
                    child{ node [minimum height=1.3em, node] {$\varepsilon$}}
                }
                child[right]{ node [minimum height=1.3em, node] {$1$}  }  
            }
            child[right]{ node [minimum height=1.3em, node] {$1$}  }  
        }
        }
        child{ node[circle, xshift=3cm]  [node] {$Z$} 
            child{ node  [circle, node] {$C_1$} 
            child[left]{node [minimum height=1.3em, node] {$a$} }
            child{ node[circle, xshift=-0.5cm] (c1) [node] {$C_1$}     
                child{ node [circle, node] {$F_1$}                   
                    child[left]{ node [minimum height=1.3em, node] {$0$} }
                    child{ node[circle, xshift=-0.5cm] [node] {$F_1$}  
                        child[left]{ node [minimum height=1.3em, node] {$0$} }
                        child{ node[xshift=-0.5cm]  [circle, node] {$F_1$}  
                            child{ node [minimum height=1.3em, node] {$\varepsilon$}}
                        } 
                    } 
                }
            }
            }
        child{ node[circle, xshift=1cm]  [node] {$C_2$}    
            child{ node[circle, xshift=0.5cm] (c2) [node] {$C_2$}    
            child{ node [circle, node] {$F_2$}   
                child{ node[circle, xshift=0.5cm] [node] {$F_2$} 
                    child{ node[circle, xshift=0.5cm] [node] {$F_2$} 
                        child{ node [minimum height=1.3em, node] {$\varepsilon$}}
                    }
                    child[right]{ node [minimum height=1.3em, node] {$1$}  } 
                }
                child[right]{ node [minimum height=1.3em, node] {$1$}  } 
            }
            }
            child[right]{ node [minimum height=1.3em, node] {$b$} }
        } 
        }
    }
    } 
;
\draw [dashed,thick,-] (a1) -- (a2) node[midway,above=9,left=-10] {10};
\draw [dashed,thick,-] (c1) -- (c2) node[midway,above=2] {12};

\end{tikzpicture}
\caption{$\prescript{}{G}\triangle_{aaa0011a0011b}$}
\label{graph:2}
\end{figure}

Let us note that dashed lines and numbers contour only denote the context dependencies,
and applied non-context-free rules, respectively, and are not part of the derivation tree. 
The pairs of context-dependent nodes are linked with dashed lines, all the other nodes are context-independent.

Since $aaa0011a0011b = \frontier(\prescript{}{G}\triangle_{aaa0011a0011b})$, all leaves are terminal nodes.
Every other node is a nonterminal node.

For a pair of neighboring paths $\alpha = SA_1A_1D_1\varepsilon$ and $\beta = SXA_2A_2A_2D_2\varepsilon$,
a string $\rho = 10$ is their context,
it is the left context of $\beta$ and the right context of $\alpha$.

\end{example}

\section{Results}

\begin{theorem}
\label{th:L:kL}
A language $L$ is $k$-linear iff there is a constant $k \geq 0$, constant $u \geq 0$ and a linear core general grammar $G$
such that $L = L(G)$
and for every $x \in L(G)$, there is a slow-branching tree of degree~$k$ denoted by
$\prescript{}{G}\triangle_x \in \prescript{}{G}\blacktriangle$ that both following satisfies:
\begin{enumerate}
    \item\label{th:L:kL:1} any two nonterminal neighboring paths contain
                           no more than $u$ pairs of context-dependent nodes;
    \item\label{th:L:kL:2} all pairs of nodes occurring in non-neighboring paths are context-independent.
\end{enumerate}
\end{theorem}

\begin{proof}\emph{Construction.}
Consider any $u \geq 0$. Let $G = (V,T,P,S)$ be a GG such that $L(G) = L$.
Set $N = V - T$. Let $P_{cs} \subseteq P$ denote the set of all non-context-free rules of $G$.
Set
\[
    N' = \{ A_{l|r} \where A \in N,\ l,r \in (P_{cs} \cup \{\eps\})^u\}.
\]
Construct a grammar $G' = (V',T,P',S_{\eps|\eps})$, where $V' = N' \cup T$.
Set $P' = \emptyset$. Construct $P'$ by performing (\ref{construction:1}) through (\ref{construction:3}) given next.
\begin{enumerate}[(I)]
    \item\label{construction:1} For all $A \ra xEy \in P$, $A \in N$, $E \in N \cup \{\eps\}$, $x,y \in T^*$, and $l,r \in (P_{cs} \cup \{\eps\})^u$,
                                if $E \in \{\eps\}$ then add $A_{\eps|\eps} \ra xy$ to~$P'$ else add $A_{l|r} \ra xE_{l|r}y$ to~$P'$;
    \item\label{construction:2} for all $A \ra BC \in P$, where $A,B,C \in N$,
                                and $r,l,x \in (P_{cs} \cup \{\eps\})^u$,
                                add $A_{l|r} \ra B_{l|x}C_{x|r}$ to $P'$;
    \item\label{construction:3} for all $p\colon AB \ra CD \in P$, $A,B,C,D \in N$,
                                $x,z \in (P_{cs} \cup \{\eps\})^u$,
                                and $y \in (P_{cs} \cup \{\eps\})^{u-1}$,
                                add $A_{x|py} \ra C_{x|y}$ and $B_{py|z} \ra D_{y|z}$ to $P'$.
\end{enumerate}

\noindent
\emph{Basic idea.} Notice nonterminal symbols.
Since every pair of neighboring paths of $G$ contains a limited number of context-dependent nodes,
all of its context-dependencies are encoded in nonterminals.
$G'$ nondeterministically decides about all context-dependencies
while introducing a new pair of neighboring paths by rules from~(\ref{construction:2}).
A new pair of neighboring paths is introduced with every application~of
\[
    A_{l|r} \ra B_{l|x}C_{x|r},
\]
where $x$ encodes a new descendant context.
Context dependencies are realized later by context-free rules from~(\ref{construction:3}).

\vspace{1em}
\noindent

Since $P'$ contains no non-context-free rule and $G'$ is context-free.
Next, we prove $L(G) = L(G')$ by establishing Claims \ref{GG:sub:kL} through \ref{GG:sup:kL}.
Define the new homomorphism $\gamma: V' \ra V$, $\gamma(A_{l|r}) = A$,
for $A_{l|r} \in N'$, and $\gamma(a) = a$ otherwise.

\begin{nclaim}
\label{GG:sub:kL}
If $S \Ra^m w$ in $G$, where $m \geq 0$ and $w \in V^*$, then $S_{\eps|\eps} \Ra^* w'$ in $G'$,
where $w' \in V'^*$ and $\gamma(w') = w$.
\end{nclaim}

\begin{proof}
We prove this by induction on $m \geq 0$.

\ibasis
Let $m = 0$. That is $S \Ra^0 S$ in $G$.
Clearly, $S_{\eps|\eps} \Ra^0 S_{\eps|\eps}$ in $G'$, where $\gamma(S_{\eps|\eps}) = S$, so the basis holds.

\ihypothesis
Suppose that there exists $n \geq 0$ such that Claim \ref{GG:sub:kL} holds for~all $0 \leq m \leq n$.

\istep
Let $S \Ra^{n+1} w$ in $G$. Then, $S \Ra^n v \Ra w$,
where $v \in V^*$, and there exists $p \in P$ such that $v \Ra w\,[p]$.
By the induction hypothesis, $S_{\eps|\eps} \Ra^* v'$, where $\gamma(v') = v$, in $G'$.
Next, we consider the following three forms of $p$.

\begin{enumerate}[(I)]
    \item Let $p\colon A \ra xEy \in P$, for some $A \in N$ , $E \in N \cup \{\eps\}$, $x,y \in T^* $.
         
          If there is no nonterminal on the right-hand side of the rule, it implies that left descendant context and a right descendant context of $A$ is $\eps$, then, by the construction of $G'$, there exists a rule $p'\colon A_{\eps|\eps} \ra xy \in P'$, where $A_{\eps|\eps}  \in v'$.
          Otherwise, suppose $l$ and $r$ are
          a left descendant context and a right descendant context of $A$.
          By the construction of $G'$, there exists a rule $p'\colon A_{l|r} \ra xE_{l|r}y \in P'$, where $A_{l|r}  \in v'$.
          Then, there exists a derivation $v' \Ra w'\,[p']$ in $G'$, where $\gamma(w') =~w$.

    \item Let $p\colon A \ra BC \in P$, for some $A,B,C \in N$.
          Without any loss of generality, suppose $l$ and $r$ are
          a left descendant context and a right descendant context of $A$,
          and $x \in (P_{cs} \cup \{\eps\})^u$ is a context of neighboring paths beginning at this node.
          By the construction of $G'$, there exists a rule $p'\colon A_{l|r} \ra B_{l|x}C_{x|r} \in P'$, where $A_{l|r}, B_{l|x}, C_{x|r} \in v'$.
          Then, there exists a derivation $v' \Ra w'\,[p']$ in $G'$, where $\gamma(w') = w$.

    \item\label{claim:1:phase:case:4} Let $p\colon AB \ra CD \in P$, for some $A,B,C,D \in N$.
          By the assumption stated in Theorem \ref{th:L:kL},
          $A$ and $B$ occur in two neighboring paths denoted by $\alpha$ and $\beta$, respectively.
          Without any loss of generality, suppose that a context of $\alpha$ and $\beta$
          is a string $c \in (P_{cs} \cup \eps)^u$, where $c = pcd$,
          and $l$ is a left descendant context, $r$ is a right descendant context of $A$, $B$, respectively.
          By the construction of $G'$, there exist two rules
          \[
              p'_l\colon A_{l|pcd} \ra C_{l|cd},\ \ p'_r\colon B_{pcd|r} \ra D_{cd|r} \in P',
          \]
          where $A_{l|pcd}, C_{l|cd}, B_{pcd|r}, D_{cd|r} \in V'$.
          Then, there exists a derivation $v' \Ra^2 w'\,[p'_lp'_r]$ in $G'$, where $\gamma(w') = w$.
\end{enumerate}

\noindent
Notice (\ref{claim:1:phase:case:4}). The preservation of the context is achieved by nonterminal symbols.
Since the stored context is reduced symbol by symbol from left to right direction
in both $\alpha$ and $\beta$, $G'$ simulates the applications of non-context-free rules of~$G$.

We covered all possible forms of $p$, so the claim holds.
\end{proof}
\begin{nclaim}
\label{pre:GG:sup:kL}
Every $x \in L(G')$ can be derived in $G'$ as follows.
\[
    S_{\eps|\eps} = x_0 \Ra^{d_1} x_1 \Ra^{d_2} x_2 \Ra^{d_3} \cdots \Ra^{d_{h-1}} x_{h-1} \Ra^{d_h} x_{h} = x,
\]
for some $h \geq 0$, where $d_i \in \{1,2\}$, $1 \leq i \leq h$, so that
\begin{enumerate}
    \item if $d_i = 1$, then $x_{i-1} = uA_{l|r}v$, $x_i = uzv$, $x_{i-1} \Ra x_i\ [A_{l|r} \ra z]$,
          where $u,v \in V'^*$,\\ $z \in \{E_{l|r},B_{l|r},C_{l|x}D_{x|r},x,y\}$,
          for some $A_{l|r},B_{l|r},C_{l|x},D_{x|r} \in  N'$, $E_{l|r} \in (N' \cup \{\varepsilon\})$,
                $x,y \in T^*$;\vspace{1em}
    \item if $d_i = 2$, then $x_{i-1} = uA_{x|py}B_{py|z}v$, $x_i = uC_{x|y}D_{y|z}v$, and
        \[
          uA_{x|py}B_{py|z}v \Ra uC_{x|y}B_{py|z}v\ [A_{x|py} \ra C_{x|y}]
          \Ra uC_{x|y}D_{y|z}v\ [B_{py|z} \ra D_{y|z}],
        \]
          for some $u,v \in V'^*$ and $A_{x|py},B_{py|z},C_{x|y},D_{y|z} \in N'$.
\end{enumerate}
\end{nclaim}

\begin{proof}
Since $G'$ is context-free, without any loss of generality in every derivation of $G'$
we can always reorder applied rules to satisfy Claim \ref{pre:GG:sup:kL}.
\end{proof}

\begin{nclaim}
\label{GG:sup:kL}
Let $S_{\eps|\eps} \Ra^{d_1} x_1 \Ra^{d_2} \cdots \Ra^{d_{m-1}} x_{m-1} \Ra^{d_m} x_{m}$
in $G'$ be a derivation that satisfies Claim \ref{pre:GG:sup:kL}, for some $m \geq 0$.
Then, $S \Ra^* w$ in $G$, where $\gamma(x_m) = w$.
\end{nclaim}

\begin{proof}
We prove this by induction on $m \geq 0$.

\ibasis
Let $m = 0$. That is $S_{\eps|\eps} \Ra^0 S_{\eps|\eps}$ in $G'$.
Clearly, $S \Ra^0 S$ in $G$. Since $\gamma(S_{\eps|\eps}) = S$, the basis holds.

\ihypothesis
Suppose that there exists $n \geq 0$ such that Claim \ref{GG:sup:kL} holds for~all $0 \leq m \leq n$.

\istep
Let $S_{\eps|\eps} \Ra^{d_1} x_1 \Ra^{d_2} \cdots \Ra^{d_{n-1}} x_{n-1} \Ra^{d_n} x_{n} \Ra^{d_{n+1}} x_{n+1}$ in $G'$
be a derivation that satisfies Claim \ref{pre:GG:sup:kL}.
By the induction hypothesis, $S \Ra^* v$, $v \in V^*$, where $\gamma(x_{n}) = v$, in $G$.
Divide the proof into two parts according to~$d_{n+1}$.

\begin{enumerate}[(A)]
    \item Let $d_{n+1} = 1$. By the construction of $G'$,
          there exists a rule $p' \in P'$ such that $x_{n} \Ra^{d_{n+1}} x_{n+1}\,[p']$.
          Next, we consider the following two forms of~$p'$.
          \vspace{1em}
          \begin{enumerate}[(I)]
              \item Let $p'\colon A_{l|r} \ra xE_{l|r}y$ or $p'\colon  A_{\eps|\eps} \ra xy  \in P'$,
                    for some $A \in N$, $E \in N$, $x,y \in T^*$ and $l,r \in (P_{cs} \cup \{\eps\})^u$.
                    By the construction of $G'$, rule $p'$ was introduced by some rule $p\colon A \ra xEy \in P$ or $p\colon A \ra xy \in P$, respectively.
                    Then, there exists a derivation $v \Ra w\,[p]$, where $\gamma(x_{n+1}) = w$.

              \item Let $p'\colon A_{l|r} \ra B_{l|x}C_{x|r} \in P'$,
                    for some $A,B,C \in N$ and $l,r,x \in (P_{cs} \cup \{\eps\})^u$.
                    By the construction of $G'$, rule $p'$ was introduced by some rule $p\colon A \ra BC \in P$.
                    Then, there exists a derivation $v \Ra w\,[p]$, where $\gamma(x_{n+1}) = w$.
          \end{enumerate}
          \vspace{1em}
    \item Let $d_{n+1} = 2$. Then, $x_{n} \Ra^{d_{n+1}} x_{n+1}$ is equivalent to
          \[
              u_1A_{x|py}B_{py|z}u_2 \Ra u_1C_{x|y}B_{py|z}u_2\ [p'_1] \Ra u_1C_{x|y}D_{y|z}u_2\ [p'_2],
          \]
          where $x_{n} = u_1A_{x|py}B_{py|z}u_2$, $x_{n+1} = u_1C_{x|y}D_{y|z}u_2$, and
          \[
              p'_1\colon A_{x|py} \ra C_{x|y},\ p'_2\colon B_{py|z} \ra D_{y|z} \in P',
          \]
          for some $u_1,u_2 \in V'^*$ and $A_{x|py}$, $B_{py|z}$, $C_{x|y}$, $D_{y|z} \in N'$.
          By the construction of $G'$, rules $p'_1$ and $p'_2$ were introduced by some rule $p\colon AB \ra CD \in P$,
          Then, there exists a derivation $v \Ra w\,[p]$, where $\gamma(x_{n+1}) = w$.
\end{enumerate}

We covered all possibilities, so the claim holds.
\end{proof}

Observe that respective the derivation trees of the constructed context-free $G'$ remain slow-branching.

\begin{nclaim}
\label{GG:kLin}
The grammar $G'$ is $k$-linear.
\end{nclaim}

\begin{proof}
In construction (\ref{construction:3}) we replace the rules of the form $AB \ra CD$ with the rules of the form $A \ra B$, where $A,B,C,D \in N$.
Therefore, only the rules that are allowed to occur in the derivation $G'$ before the rules of the form $A \ra BC$ are the rules of the form $A \ra B$.
Rules of the form $A \ra B$  before the rules of the form $A \ra BC$ can be omitted by the trivial transformation of $G'$, similar to the algorithm on elimination of unit productions from Section 5 in \cite{Med20}.
Therefore, the grammar $G'$ is $k$-linear.
\end{proof}

By Claim \ref{GG:kLin} $G'$ is $k$-linear.
By Claims \ref{GG:sub:kL} and \ref{GG:sup:kL},
$S \Ra^* w$ in $G$ iff $S_{\eps|\eps} \Ra^* w'$ in $G'$, where $\gamma(w') = w$.
If $S \Ra^* w$ in $G$ and $w \in T^*$, then $w \in L(G)$.
Since $\gamma(w') = w' = w$, for $w \in T^*$, $w' \in L(G')$.
Therefore, $L(G) = L(G')$ and Theorem \ref{th:L:kL} hold.
\end{proof}

Consider Theorem \ref{th:L:kL}. Observe that the \ref{th:L:kL:2}nd condition is superfluous whenever $G$ is propagating.

\begin{theorem}
\label{th:L:kL:P}
A language $L$ is $k$-linear iff there is a constant $k \geq 0$, constant $u \geq 0$ and a propagating linear core general grammar $G$
such that $L = L(G)$ and for every $x \in L(G)$,
there is a slow-branching tree of degree $k$ $\triangle_x \in \prescript{}{G}\blacktriangle$,
where any two nonterminal neighboring paths contain no more than $u$ pairs of context-dependent nodes.
\end{theorem}

\begin{proof}
Prove this by analogy with the proof of Theorem \ref{th:L:kL}.
\end{proof}

\begin{theorem}
\label{th:L:R}
A language $L$ is regular iff there is a constant $u \geq 0$ and a left linear core general grammar $G$ such that $L = L(G)$
and for every $x \in L(G)$, there is a tree 
$\triangle_x \in \prescript{}{G}\blacktriangle$ that satisfies:
\begin{enumerate}
    \item\label{th:L:R:1} any two nonterminal neighboring paths contain
                           no more than $u$ pairs of context-dependent nodes;
    \item\label{th:L:R:2} out of neighboring paths, any pair of nodes is context-independent.
\end{enumerate}
\end{theorem}

\begin{proof}
Prove this by analogy with the proof of Theorem \ref{th:L:kL}.
\end{proof}

\begin{theorem}
\label{th:L:R:P}
A language $L$ is regular iff there is a constant $u \geq 0$ and a propagating left linear core general grammar $G$
such that $L = L(G)$ and for every $x \in L(G)$,
there is a tree $\triangle_x \in \prescript{}{G}\blacktriangle$,
where any two nonterminal neighboring paths contain no more than $u$ pairs of context-dependent nodes.
\end{theorem}

\begin{proof}
Prove this by analogy with the proof of Theorem \ref{th:L:kL}.
\end{proof}

\section{Use}

In this section, we explain how to apply the results achieved in the previous section in order to demonstrate the metalinearness (or regularity) of a language, $L$. As a rule, this demonstration follows the next three-step proof scheme for metalinearness.

\begin{enumerate}
    \item Construct a linear core GG $G$.
    \item Prove $L(G) = L$.
    \item Prove that $G$ satisfies conditions from Theorem \ref{th:L:kL:2} or Theorem \ref{th:L:kL} depending on whether $G$ is context-sensitive.
\end{enumerate}

For regularity, we use a similar three-step scheme as following.

\begin{enumerate}
    \item Construct a left linear core GG $G$.
    \item Prove $L(G) = L$.
    \item Prove that $G$ satisfies conditions from Theorem \ref{th:L:R} or Theorem \ref{th:L:R:P} depending on whether $G$ is context-sensitive.
\end{enumerate}

Reconsider the grammar $G$ from Example \ref{ex:GG}. Following the proof scheme sketched above,  we next prove that $L(G)$ is $k$-linear.
Without any loss of generality, every terminal derivation of $G$ can be divided into the following 5 phases,
where each rule may be used only in a specific phase:

\noindent
\begin{center}
\begin{inparaenum}[(a)]
    \item\label{phase:a} \ref{rule:1}--\ref{rule:4}
    \item\label{phase:b} \ref{rule:5}--\ref{rule:9}
    \item\label{phase:c} \ref{rule:10}--\ref{rule:12}
    \item\label{phase:d} \ref{rule:13}--\ref{rule:17}
    \item\label{phase:e} \ref{rule:18}--\ref{rule:22}
\end{inparaenum}
\end{center}

Next, we describe these phases in greater detail. 

\begin{enumerate}[(a)]
    \item First, we generate the following string by rules \ref{rule:1} though  \ref{rule:4}.
          \[
            A_1A_2BC_1C_2
          \]
          Possibly applicable rules from (\ref{phase:b}) and (\ref{phase:c}) may be postponed to the next phases without affecting the derivation, since the rules in the previous phases cannot rewrite the nonterminals of the following phases.
          \vspace{1em}
    \item The rules (\ref{rule:5}) through (\ref{rule:9}) are context-free rules and nonterminals on the left-hand side of the rule are the same as on the right-hand side of the rule. Therefore, they are grouped into (\ref{phase:b}), since they only generate terminals. Possibly applicable rules from (\ref{phase:c}) may be postponed for the phase (\ref{phase:c})
          without affecting the derivation since the rules in the previous phases cannot rewrite nonterminals from the following phases.
          \[
            a^*A_1A_2a^*b^*Bc^*a^*C_1C_2b^*.
          \]
    \item The rules \ref{rule:10} and \ref{rule:12} are non context-free rules. The rules \ref{rule:10} through \ref{rule:12} are all rules without generating terminals. For the same reason as in (\ref{phase:a}) rules \ref{rule:1} to \ref{rule:4} from the phases (\ref{phase:d}) and (\ref{phase:e}) can be postponed to respective phases.
          \[
            a^*D_1D_2a^*b^*Ec^*a^*F_1F_2b^*.
          \]
    \item The rules \ref{rule:13} through \ref{rule:17} are alike rules in (\ref{phase:b})
          \[
            a^*0^*D_1D_21^*a^*b^*0^*E1^*c^*a^*0^*F_1F_21^*b^*.
          \]
    \item Since rules \ref{rule:18} and \ref{rule:22} are erasing rules and
          they can always be postponed until the end of any successful derivation.
          \[
            a^*0^*1^*a^*b^*0^*1^*c^*a^*0^*1^*b^*.
          \]
\end{enumerate}

Grammar $G$ is obviously a linear core GG.

Only rules in the step (\ref{phase:a}) include branching of nonterminals, no terminals are generated and the branching in the step (\ref{phase:a}) is a slow-branching since the degree derivation tree is $4$ and, therefore, $u$ is always $4$.
Therefore, the slow-branching condition is fulfilled.

Let us now show that for any $x \in L(G)$, there is $\prescript{}{G}\triangle_x \in \prescript{}{G}\blacktriangle$,
where any two neighboring paths contain no more than a one pair of context-dependent nodes.

Every pair of context-dependent nodes in $\prescript{}{G}\triangle_x$
corresponds to one non-context-free rule in $S \Ra^* x$.
Consider the five phases sketched above. Observe that all phases except (\ref{phase:c})
contain only non context-free rules, so we only have to investigate (\ref{phase:c}).
On the other hand, (\ref{phase:c}) contain no rule of the form $A \ra BC$,
thus the number of neighboring paths remains unchanged.

In (\ref{phase:c})  rule \ref{rule:10} and \ref{rule:12} introduce context dependency between two pairs of neighboring paths. After the application of these two rules, we cannot reach the nonterminals again on the left-hand side of rules \ref{rule:10} and \ref{rule:12}. Therefore, these context-dependencies can occur only once between a pair of neighboring paths.

No other non-context-free rule is applied; therefore, no other context-dependent pair of nodes can occur.
Then, every pair of neighboring paths may contain at most one context-dependent pair of nodes introduced in phase
(\ref{phase:c}).

Since $G$ is a linear core GG, where for every $x \in L(G)$,
there is $\prescript{}{G}\triangle_x \in \prescript{}{G}\blacktriangle$,
where any two neighboring paths contain no more than one pair of context-dependent nodes,
by Theorem \ref{th:L:kL}, $L(G)$ is $k$-linear.

Unfortunately, although we are able to transform any GG into KNF and, that is, linear core GG, the question whether the conditions in Theorems \ref{th:L:kL} through \ref{th:L:R:P} are satisfied is obviously undecidable. To invent an algorithm that gives at least approximate results is part of a future research.

\section{Final Remarks and Open Problems}

Before closing this paper, we bring the reader's attention to an open question. More specifically, consider a more lenient definition of slow-branching tree as follows.

\begin{ndefinition}
\label{def:alt}
A labeled ordered tree $t$ is \emph{slow-branching} if any of its pairs of nonterminal neighboring paths
contains no more than two nonterminal nodes having two nonterminal children. 
A slow-branching labeled ordered tree is of \emph{degree $k$}
if it contains $k$ branching nonterminal nodes, $k \geq 1$.
\end{ndefinition}

It is obvious that the newly provided Definition \ref{def:alt} is insufficient to prove that a grammar restricted by a slow-branching derivation tree is $k$-linear. However, it is possible to apply different restrictions to  Definition \ref{def:alt} with its own advantages or demonstrate similar result to Theorem \ref{th:L:kL} to prove that it is $k$-linear. Such a discovery would require further studies.

\section*{Acknowledgments}
This work was supported by the BUT grant FIT-S-23-8209.

\nocite{*}
\bibliographystyle{eptcs}
\bibliography{generic}
\end{document}